\documentclass[11pt]{article}

\usepackage{xspace}
\usepackage{verbatim}
\usepackage{color}
\usepackage{graphicx}
\usepackage{tabularx}

\usepackage{amsmath}
\usepackage[showonlyrefs]{mathtools}
\usepackage{mathstyle}
\usepackage{breqn}
\usepackage{empheq}
\usepackage{amstext,amssymb,amsfonts}
\usepackage{fullpage}
\usepackage{nicefrac}
\usepackage{bm}

\usepackage{nameref}
\usepackage[pagebackref,colorlinks,linkcolor=blue,filecolor = blue,
citecolor = blue, urlcolor = blue]{hyperref}
\usepackage{cleveref}

\usepackage{textcomp,setspace}

\newcommand{\nfrac}{\nicefrac}

\newcommand{\defeq}{\stackrel{\mathrm{def}}=}

\newcommand{\poly}{\mathrm{poly}}

\newcommand{\Paren}[1]{\left(#1 \right )}

\newcommand{\Brac}[1]{\left[#1 \right]}
\newcommand{\set}[1]{\{#1\}}

\newcommand{\Set}[1]{\left\{#1\right\}}
\newcommand{\abs}[1]{\lvert#1\rvert}
\newcommand{\Abs}[1]{\left\lvert#1\right\rvert}

\newcommand{\norm}[1]{\lVert#1\rVert}

\definecolor{DSred}{rgb}{1,0,0}
\newcommand{\Authornote}[2]{{\small\textcolor{DSred}{\sf$<<<${  #1: #2 }$>>>$}}}

\renewcommand{\leq}{\leqslant}
\renewcommand{\geq}{\geqslant}
\renewcommand{\ge}{\geqslant}
\renewcommand{\le}{\leqslant}
\renewcommand{\epsilon}{\varepsilon}
\newcommand{\eps}{\epsilon}

\usepackage{amsthm}
\usepackage{thmtools}

\declaretheorem[within=section]{theorem}

\declaretheorem[sibling=theorem]{lemma}
\declaretheorem[sibling=theorem]{claim}

\declaretheorem[sibling=theorem]{definition}

\declaretheorem[sibling=theorem]{remark}

\newcounter{termcounter}
\renewcommand{\thetermcounter}{\Alph{termcounter}}
\crefname{term}{term}{terms}
\creflabelformat{term}{\textup{(#2#1#3)}}

\makeatletter
\def\term{\@ifnextchar[\term@optarg\term@noarg}
\def\term@optarg[#1]#2{%
  \textup{(#1)}%
  \def\@currentlabel{#1}%
  \def\cref@currentlabel{[][2147483647][]#1}%
  \cref@label[term]{#2}}
\def\term@noarg#1{%
  \refstepcounter{termcounter}%
  \textup{(\thetermcounter)}%
  \cref@label[term]{#1}}
\makeatother

\newcommand{\R}{\mathbb{R}}
\newcommand{\C}{\mathbb{C}}

\newcommand{\N}{\mathbb{N}}

\newcommand{\cA}{\mathcal A}
\newcommand{\cB}{\mathcal B}

\newcommand{\Esymb}{{\bf E}}
\newcommand{\Isymb}{{\bf I}}
\newcommand{\Psymb}{{\bf P}}

\DeclareMathOperator*{\E}{\Esymb}

\DeclareMathOperator*{\ProbOp}{\Psymb}
\DeclareMathOperator*{\I}{\Isymb}
\renewcommand{\Pr}{\ProbOp}

\newcommand{\ComplexityFont}[1]{\ensuremath{\mathsf{#1}}}

\newcommand{\NP}{\ComplexityFont{NP}}

\newcommand{\BPP}{\ComplexityFont{BPP}}

\newcommand{\HashP}{\ComplexityFont{\#P}}        

\newcommand{\ptest}{\ensuremath{\mathsf{PTest}}\xspace}
\newcommand{\reject}{\texttt{Reject}\xspace}

\newcommand{\lintest}{$\mathsf{LinearityTest}$}
\newcommand{\tailtest}{$\mathsf{TailTest}$}

\newcommand{\oracle}{\ensuremath{\mathcal{O}}}
\newcommand{\stack}{\ensuremath{\set{{\oracle}_k}_{\set{k \leq n}}}\xspace}

\newcommand{\Perm}{\ensuremath{\mathsf{Per}}}
\newcommand{\per}{\Perm}
\newcommand{\perm}{\Perm}

\newcommand{\err}{\mathsf{err}_{2}}
\newcommand{\serr}{\mathsf{err}_{2,\eta}}
\newcommand{\etanorm}[2]{\ensuremath{\norm{#1}_{2,#2}}}

\newcommand{\cgaussian}%
 {\ensuremath{\mathcal{N}(0,1)_{\C}^{k \times k}}}

\newcommand{\Babs}[1]{\Big\lvert#1\Big\rvert}
\newcommand{\BBrac}[1]{\Big[#1\Big]}
\newcommand{\BNorm}[1]{\Big\lVert#1\Big\rVert}

\newcommand{\ignore}[1]{}

\title{Testing Permanent Oracles -- Revisited}

\author{Sanjeev Arora\thanks{Princeton University, Computer Science
    Department and Center for Computational Intractability. This work
    is supported by the NSF grants CCF-0832797 and
    CCF-1117309. Email:~\texttt{\{arora,rajsekar,sachdeva\} @cs.princeton.edu}.}
  \and Arnab Bhattacharyya\thanks{Princeton University, Computer Science
    Department and Center for Computational Intractability. This work
    is supported by NSF Grants CCF-0832797,
    0830673, and 0528414. Email: \texttt{arnabb@cs.princeton.edu}.}
  \and Rajsekar Manokaran\footnotemark[1] \and Sushant
  Sachdeva\footnotemark[1]} \date{}

\newcommand{\Rnote}[1]{\Authornote{Rajsekar}{#1}}
\newcommand{\Anote}[1]{\Authornote{Arnab}{#1}}

\begin{document}
\maketitle

\begin{abstract}
Suppose we are given an oracle that claims to approximate the
permanent for most matrices $X$, where $X$ is  chosen from the
  {\em  Gaussian ensemble} (the  matrix entries  are i.i.d.~univariate
  complex Gaussians).  Can we test that the  oracle satisfies this
  claim?  This paper gives a polynomial-time algorithm for the task.

  The oracle-testing problem is of interest because a recent paper of
  Aaronson  and Arkhipov  showed that  if there  is  a polynomial-time
  algorithm  for   simulating  boson-boson  interactions   in  quantum
  mechanics, then  an approximation oracle  for the permanent  (of the
  type  described   above)  exists  in   $\BPP^\NP$.  Since  computing
  the permanent of  even $0/1$ matrices is  $\HashP$-complete, this seems
  to demonstrate  more computational  power in quantum  mechanics than
  Shor's factoring algorithm does. However, unlike factoring, which is
  in \NP, it was unclear previously how to test  the   correctness
  of  an  approximation   oracle  for  the   permanent, and this is
  the contribution of the paper. 

  The technical difficulty overcome here is that univariate polynomial
  self-correction, which underlies similar oracle-testing algorithms
  for permanent over {\em finite fields} ---and whose discovery led to
  a  revolution in complexity  theory---does not  seem to  generalize to
  complex (or even, real) numbers.   We believe that
  this  tester will  motivate further  progress on  understanding the
  permanent of Gaussian matrices.
\end{abstract}

\section{Introduction}
The permanent of an $n$-by-$n$ matrix
$X = (x_{i,j})$ is defined as
$$\Perm(X) = \sum_\pi \prod_{i=1}^n x_{i,\pi(i)},$$
where $\pi$ ranges over all permutations from $[n]$ to $[n]$.
A recent paper of Aaronson and Arkhipov \cite{AaronsonA11}~(henceforth
referred to as AA) introduced a surprising connection between quantum
computing and the complexity of computing the permanent (which is
well-known to be \HashP-complete to compute in the worst
case~\cite{Valiant77}). They define and study a formal model of
quantum computation with non-interacting bosons in which $n$ bosons
pass through a ``circuit'' consisting of optical elements. Each boson
starts out in one of $m$ different phases and, at the end of the
experiment, the system is in a superposition of the basis states---one for each
possible partition of the $n$ bosons into $m$ phases.

AA proceed to show that if there is an efficient classical randomized
algorithm $\cA$ that simulates the experiment, in the sense of being
able to output random samples from the final distribution (up to a
small error in total variation distance) of the Bosonic states at the
end of the experiment, then there is a way to design an {\em
 approximation} algorithm $\cB$ in $\BPP^\NP$ for the permanent
problem for an interesting family of random matrices. The random
matrices are drawn from the Gaussian ensemble---each entry is an
independent standard Gaussian complex number---and the algorithm
computes an additive approximation, in the sense that,
\begin{equation}
 \Abs{\cB(X) - \Perm(X)}^2 \le \delta^2 n!,\label{eq:add-approx}
\end{equation} 
for at least a fraction $1 - \eta$ of the input matrices $X$. (Note
that the {\em variance} of $\Perm(X)$ is $n!$ for Gaussian ensembles,
so this approximation is nontrivial.) The running time of $\cB$ is
$\poly(n, \nfrac{1}{\delta}, \nfrac{1}{\eta})$ with access to an
oracle in $\NP^\cA$. In other words, $\cB \in \BPP^{\NP^\cA}$ for
$\eta, \delta = \Omega\Paren{\nfrac{1}{\poly(n)}}$ (refer to Problem $2$ and
Theorem $3$ in \cite{AaronsonA11}). The authors go on to conjecture that
obtaining an additive approximation as in \cref{eq:add-approx} is
$\HashP$-hard (this follows from Conjectures $5$ and $6$, and Theorem
$7$ in \cite{AaronsonA11}). If true, this conjecture has surprising
implications for the computational power of quantum systems. By
contrast, the crown jewel of quantum computing, Shor's
algorithm~\cite{Shor94}, implies that the ability to simulate quantum
systems would allow us to factor integers in polynomial time, but
factoring (as well as other problems known to be in BQP) is not even
known to be $\NP$-Hard.

As evidence for their conjecture, Arkhipov and Aaronson point to
related facts about the permanent problem for matrices over integers
and finite fields. It is
known that that if there is a constant factor approximation algorithm
for computing $\Perm(X)$ where $X$ is an arbitrary matrix of integers,
then one can solve $\HashP$ problems in polynomial time. Thus,
approximation on {\em all} inputs seems difficult\footnote{Note that
 approximating the permanent is known to be feasible for the special
 case of non-negative real matrices \cite{Broder,JerrumSinclair,JerrumSinclairVigoda}.}. Likewise, starting with a paper of
Lipton, researchers have studied the complexity of computing the
permanent (exactly) for {\em many matrices}. For example, given an
algorithm that computes the permanent exactly for $1/\poly(n)$
fraction of all matrices $X$ over a finite field $GF(p)$ (where $p$ is
a sufficiently large prime), one can use self-correction procedures
for univariate polynomials~\cite{Gem,GemSud,Cai} to again obtain
efficient randomized algorithms for \HashP-hard problems.

Thus, either restriction ---approximation on all matrices, or the
ability to compute exactly on a significant fraction of matrices---
individually results in a \HashP-hard problem. What makes the AA
conjecture interesting is that it involves the {conjunction} of the
two restrictions: the oracle in question {\em approximates} the value
of the permanent for {\em most} matrices.

The focus of the current paper is the following question: {\em given
 an additive approximation oracle for permanents of Gaussian
 matrices ($\cB$ in 
\cref{eq:add-approx} above), how can we test that the oracle is
correct? } We want a tester that accepts with high probability when
$\cB$ satisfies the condition in \cref{eq:add-approx} and rejects with
high probability when $\cB$ does not approximate well on a substantial
fraction of inputs. Note that the testing problem is a non-issue for
previous quantum algorithms such as Shor's algorithm, since the
correctness of a factoring algorithm is easy to test.

The testing question has been studied for the permanent problem over
finite fields. Given an oracle that supposedly computes $\perm(\cdot)$
for even, say, $\nfrac{3}{4}^{\textrm{th}}$ of the matrices over
$GF(p)$, one can verify this claim using self-correction for
polynomials over finite fields and the {\em downward
  self-reducibility} of $\perm(\cdot)$, as described below in more
detail in Section \ref{sec:related}. (In fact, if the oracle satisfies
the claim, then one can compute $\perm(\cdot)$ on all matrices with
high probability.)  However, as noted in AA, these techniques that
work over finite fields fail badly over the complex numbers. The
authors in AA also seem to suggest that techniques
analogous to self-correction and downward self-reducibility can be
generalized to complex numbers in some way, but this remains open.

In this paper, we solve the testing problem using downward
self-reducibility alone. Perhaps this gives some weak evidence for
the truth of the AA conjecture. Note that since we lack
self-correction techniques, we do not get an oracle at the end that
computes the permanent for all matrices as in the finite field case.
Incidentally, an argument similar to the one presented in this paper
works in the finite field case also, giving an alternate tester for
the permanent that does not use self-correction of polynomials
over finite fields. 
\subsection{Related Work}\label{sec:related}
As mentioned above, testing an oracle for the permanent over finite fields has been
extensively studied. The approach, basically arising from
\cite{LFKN}, uses self-correction of polynomials over finite fields
and downward self-reducibility of the permanent. Let us revisit the
argument.

Suppose we are given a sequence of oracles $\Set{\oracle_k}_k,$ where
for each $k$, $\oracle_k$ allegedly computes the permanent for a
$9/10$ fraction of all $k$-by-$k$ matrices over the field. The
argument proceeds by first applying a self-correction procedure for
low-degree polynomials (see \cite{GemSud}), noting that the permanent
is a $k$-degree multilinear polynomial in the $k^2$ entries of the
matrix, treated as variables.

The correction procedure, on input $X,$ queries $\oracle_k$ at
$\poly(n)$ points, and outputs the correct value of $\perm(X)$ with $1
- \exp(-n)$ probability (over the coin tosses of the procedure). Thus,
the procedure acts as a proxy for the oracle, providing
$\Set{\oracle^\star_k}_k$ which can now be tested for mutual
consistency using the downward self-reducibility of the permanent:
\begin{equation}
 \perm(X) = {\textstyle \sum_j} x_{1,j} \cdot \perm(X_{j}).\label{eq:self-red}
\end{equation}
Here, $X_j$ is the submatrix formed by removing the first row and
$j^{\textrm{th}}$ column. Finally, since $\oracle_1$ can be verified
by direct computation, this procedure tests and accepts sequences
where $\oracle_k$ computes the permanent of a fraction $9/10$ of all
$k \times k$ matrices; while rejecting sequences of oracles where for
some $k$, $\oracle_k(X) \neq \per_k(X)$ on more than, say a fraction
$3/10$, of the inputs.

A natural attempt to port this argument to real/complex gaussian
matrices runs into fatal issues with the self-correction procedures:
since the oracles are only required to {\em approximate} the value of
the permanent, a polynomial interpolation procedure incurs an
exponential (in the degree) blow-up in the error at the point of
interest (see \cite{Arora-Khot}). In our work, we circumvent
polynomial interpolation and only deal with self-reducibility, noting
that \cref{eq:self-red} expresses the permanent as a {\em linear}
function of permanent of smaller matrices.

\subsection{Overview of the Tester}
%
We work with the following notion of quality of an oracle, naturally
inspired by the AA conjecture: the approximation guarantee achieved by the oracle
on all but a small fraction of the inputs.
\begin{definition} For an integer $n$, an oracle $\oracle_n:
  \C^{n\times n} \to \C$, is said to be $(\delta, \eta)$-good, if, an
  $n \times n$ matrix $X$ sampled from the Gaussian ensemble satisfies
  $\Abs{\oracle_n(X) - \perm_n(X)}^2 \leq \delta^2 n!,$ with probability
  at least $1-\eta$ over the sample.
\end{definition}

Note that since the tester is required to be efficient, we
(necessarily) allow even good oracles to answer arbitrarily on a small
fraction of inputs, because the tester will not encounter these bad
inputs with high probability. As an aside, there is also the issue of
{\em additive} vs {\em multiplicative} approximation, which AA
conjecture have similar complexity. In this paper, we stick with
additive approximation as defined above.

Our main result is stated
informally below~(see Theorem~\ref{thm:main} for a precise statement).
\begin{theorem}[Main theorem -- informal]\label{thm:main-informal}
  There exists an algorithm $\cA$ that, given a positive integer $n,$
  an error parameter{\footnote{All of the $\poly(\cdot)$ are fixed polynomials, hidden
      for clarity}} $\delta \geq \nfrac{1}{\poly(n)}$, and access to
  oracles $\{\oracle_k\}_{1\leq k\leq n}$ such that $\oracle_k : \C^{k^2} \to
  \C$, has the following behavior:
\begin{itemize}
\item If for every $k \leq n$, the oracle $\oracle_k$ is
 $\left(\delta, \nfrac{1}{\poly(n)}\right)$-good, then $\cA$ accepts
 with probability at least $1-\nfrac{1}{\poly(n)}$.
\item If there exists a $k \leq n$ such that the oracle $\oracle_k$ is
 not even $\Paren{ \poly(n)\cdot \delta,
 \nfrac{1}{\poly \Paren{n}} }$-good, then $\cA$ rejects with
 probability at least $1-\nfrac{1}{\poly(n)}$.
\item
 The query complexity as well as the time complexity of $\cA$ is
 $\poly(n/\delta)$. 
\end{itemize}
\end{theorem}

We conduct the test in $n$ stages, one stage for each submatrix
size. Let $k \leq n$ denote a fixed stage, and let $X \in
\C^{k^2}$. Now, using downward self-reducibility~(\cref{eq:self-red}),
we have,
\begin{align}
& \textstyle  \Abs{\oracle_k(X) - \per_k(X)} \le 
  \underbrace{\Babs{\oracle_k(X) -
 {\textstyle \sum_j} x_j \oracle_{k-1}(X_j)}}_{\term{trm:A}} + 
  \underbrace{\Babs{{\textstyle \sum_j} x_j \Brac{\oracle_{k-1}(X_j) -
 \per_{k-1}(X_j)}}}_{\term{trm:B}}.\label{eq:error-bound}
\end{align}
Recall that $X_j$ is the submatrix formed by removing the first row
and $j^{\textrm{th}}$ column (often referred to as a minor).  

We bound \cref{trm:A} above, by checking if $\oracle_k$ is a linear
function in the variables along the first row ($x_j$ in above), when
the rest of the entries of the matrix are fixed; the coefficients of
the linear function are determined by querying $\oracle_{k-1}$ on the
$k$ minors along the first row. The tolerance needed in the test is
estimated as follows: a good collection of oracles estimates
$\per_{k-1}$ up to $\delta \sqrt{(k - 1)!},$ and $\per_k$ up to
$\delta \sqrt{k!}$ additive error. Further, since the expression is
identically zero for the permanent function, we have:
\begin{align*}
\eqref{trm:A} &\le \Abs{\oracle_k(X) - \per_k(X)}%
 + \Babs{ {\textstyle \sum_j} x_j \Paren{\oracle_{k-1}(X_j)-\per_{k-1}(X_k)}}\\
 &\le \delta \sqrt{k!}%
 + \Babs{{\textstyle \sum_j} x_j \delta \sqrt{(k-1)!}}\le \delta\sqrt{k!} \cdot (1 + O(\sqrt{\log n})) ,
\end{align*}
where the last inequality follows from standard Gaussian tail bounds.

We test this by simply querying the oracles for random $X$ and the
minors obtained thereof and checking if the downward self-reducibility
condition is approximately met.

The second term, \cref{trm:B}, is linear in the error $\oracle_{k-1}$
makes on the minors, say $\eps_{k-1} \sqrt{(k-1)!}$ on each minor. A
naive argument as above says term $(B)$ is at most $\eps_{k-1}
\sqrt{k!} \cdot \Theta(\sqrt{\log n})$. From this and
\cref{eq:error-bound}, the error in $\oracle_k$ is at most a $\Theta(\sqrt{\log n})$
factor times the error in $\oracle_{k-1}$. However, this bound is too
weak to conclude anything useful about $\oracle_n$.

We overcome this issue by measuring the error in a
root-mean-square~(RMS or $\ell_2$) sense as follows:
\[ \err(\oracle_k) = \sqrt{\E_X\Brac{\oracle_k(X) - \per_k(X)}^2} =
\norm{\oracle_k - \per_k}_2 .\]
Now, 
\[ \norm{\oracle_k - \per_k}_2 \le \norm{\oracle_k - {\textstyle
    \sum_j} x_j \oracle_{k-1}(X_j)}_2 + \sqrt{ \E \Big[{{\textstyle
      \sum_j} x_j ( \oracle_{k-1} - \per_{k-1} )}\Big] ^2 }.\] The
first term is still $\delta \sqrt{k!} \cdot O(\sqrt{\log n})$ assuming
the linearity test passes. Since each $x_i$ is an independent standard
Gaussian, the second term is at most $\sqrt{k} \cdot
\err(\oracle_{k-1}) = \eps_{k-1} \cdot \sqrt{k!}$. Then, $
\err(\oracle_k) \le (\delta \sqrt{\log n} + \eps_{k-1}) \cdot
\sqrt{k!} ,$ and thus $\err(\oracle_n)$ is at most $\poly(n) \delta
\sqrt{n!}$ as we set out to prove! The caveat however is that $\err$
as defined cannot be bounded precisely because we necessarily need to
discount a small fraction of the inputs: the oracles could be
returning arbitrary values on a small fraction, outside the purview of
any efficient tester. We deal with this by using a more sophisticated
RMS error that discounts an $\eta$-fraction of the input:
\[ \serr(\oracle_k) = \inf_{S: \mu(S) \le \eta} \sqrt{ \E_X\Brac{ 1_s
 (\oracle_k(X) - \per_k(X))}^2},\]
where $1_S$ denotes the indicator function of the set $S.$
We then use a tail inequality on the permanent based on its fourth
moment to carry through the inductive argument set up above. This
requires a {\em Tail Test} on the oracles to check that the oracles
have a tail similar to the permanent. Our analysis shows that the
Linearity and Tail test we design are sufficient and efficient,
proving~\Cref{thm:main-informal}.

\vspace{.25em}
\noindent \emph{Organization.}
In the next section, we
set up the notation. \Cref{sec:ptester} describes the test we design
and follows it up with its analysis.

\section{Preliminaries}
\paragraph{Notation and Setup.}
  We deal with complex valued functions
on the space of square matrices over the complex numbers, $\C^{k
  \times k}$ for some integer $k$. We assume $\C^{k \times k}$ is
endowed with the standard Gaussian measure $\cgaussian.$ We use the
notation $\Pr_X[E]$ to denote the probability of an event $E,$ when $X
\sim \cgaussian.$ We denote by $\E_X[Y]$ to denote the expectation of
the random variable $Y,$ when $X \sim \cgaussian.$

Functions from $\C^d$ to $\{0, 1\}$ are called indicator functions
(since they indicate inclusion in the set of points where the
function's value is $1$). We denote the indicator function for a
predicate $q(X)$ by $\I[q(X)]$ and define it to be $1$ when $q(X)$ is
true and $0$ otherwise. 
For example, $\I[\abs{x} \ge 2]$ is $1$ for all
$x$ whose magnitude is at least $2,$ and $0$ otherwise.

\paragraph{Error and $\ell_2$ norm of Oracles.}
 The (standard)
$\ell_2$ norm of a square-integrable function $f: \C^d \to \C$ is
denoted by $\norm{f}_2$ and is equal to $\E_X [\Abs{f}^2],$ where $X
\sim \mathcal{N}(0,1)_{\C}^{d}.$  
\ignore{
We also
define a variant, $\etanorm{\cdot}{\eta}$ as follows:
\begin{equation}\label{eq:eta-norm-defn}
\etanorm{f}{\eta} \defeq \inf_{S: \mu(S) \ge 1 - \eta} \sqrt{ \E_X
  \Brac{ 1_S f^2(X) }}.
\end{equation}
Strictly speaking, $\etanorm{\cdot}{\eta}$ is not a norm since it does
not satisfy sub-additivity.  However, the following approximation to
it will suffice for our purposes:
\begin{equation}\label{eq:eta-norm-subadditive}
\etanorm{f+g}{\eta_1 + \eta_2} \le \etanorm{f}{\eta_1} +
\etanorm{f}{\eta_2}.
\end{equation}}
An oracle for the permanent is simply a function $\oracle_k:
\C^{k\times k} \to \C$ that can be queried in a single time unit.  We
will work with a sequence of oracles \stack, one for every dimension
$k$ less than $n.$ 
\ignore{In the notation set up above, the RMS error of an
oracle, $\oracle_k$, while discounting an $\eta$-fraction is:
\begin{equation}\label{eq:rms-sophisticated}
\serr(\oracle_k) \defeq \etanorm{\oracle_k - \per_k}{\eta}.
\end{equation}}

\paragraph{Moments of Permanents.}
 The first and the second moments of
the permanent under the Gaussian distribution on $k\times k$ matrices
are easy to compute:
$\E_X[\Perm_k(X)] = 0,\ \E_X[|\Perm_k(X)|^2] = k!.$ We also know the
fourth moment of the permanent function for Gaussian matrices,
$\E_X[|\Perm_k(X)|^4] = (k+1)(k!)^2$ (Lemma
56, \cite{AaronsonA11}). This fact and Markov's
inequality immediately imply:
\begin{lemma}[Tail Bound for Permanent] 
\label{lem:perm-tail-bound}
For every positive integer $k,$ the permanent satisfies
$\Pr_X[|\Perm_k(X)| > T\sqrt{k!}] \le \nfrac{(k+1)}{T^4}.$
\end{lemma}

\section{Testing Approximate Permanent Oracles}
\label{sec:ptester}

Our testing procedure, \ptest, has three parameters: a positive
integer $n$, the dimension of the matrices being tested; $\delta \in
(0,1]$, the amount of error allowed; and $c \in (0,1]$, a completeness
parameter\footnote{We require the mild condition that $n =
  \Omega\left(\sqrt{\log \frac{1}{\delta c}}\right),$ which is
  satisfied for large enough $n$ when $c,\delta =
  \frac{1}{\poly(n)}.$}. In addition, it has query access to the sequence
of oracles, \stack being tested. In the following, for a matrix $X,$
we denote the entries in the first row of $X$ by
$x_{11},\ldots,x_{1k},$ and by $X_i$ the minor obtained by removing
the first row and the $i^{\textrm{th}}$ column from $X.$ (There will
be no confusion since we will only be working with expansion along the
first row.)

The guarantees of the tester are twofold: it accepts with probability
at least $1-c,$ if, for every $k$, and every $X \in \C^{k\times k},$ we
have $\Abs{\oracle_k(X) - \per_k(X)}^2 \le \delta^2 k!$; on
the other hand, the tester almost always rejects if for some $k \le
n$, $\oracle_k(X)$ is not $\poly(n)\delta\cdot \sqrt{k!}$ close to
$\per_k(X)$ with probability $1-\frac{1}{\poly(n)}$ over $X$ (see
below for precise theorems).  The query complexity of \ptest is
bounded by $\poly (n, \nicefrac{1}{\delta}, \nicefrac{1}{c}).$
Assuming that each oracle query takes constant time, the time
complexity of \ptest is also bounded by $\poly(n,
\nicefrac{1}{\delta}, \nicefrac{1}{c})$ (see below for precise
bounds).

The test consists of two parts: The first is a \emph{linearity} test,
that tests that the oracles \stack satisfy $\oracle_k(X) \approx
\sum_i x_{1i}\oracle_{k-1}(X_i)$ (observe that the permanent satisfies
this exactly). The second part is a \emph{tail} test, that tests that
the function does not take large values too often (the permanent
satisfies this property too, as shown by
Lemma~\ref{lem:perm-tail-bound}).

\begin{quote}
{\bf \lintest$(n,k,\delta)$}: 
Sample a $k \times k$ matrix $X \sim \cgaussian.$ If $k=1,$ output
\reject unless $\left|\oracle_k(X) - X\right|^2 \le n^2\cdot
\delta^2.$ Else, test if:
\[ \Babs{\oracle_k(X) - \sum_{i=1}^k x_{1i} \oracle_{k-1}(X_i)}^2 \le
n^2 \delta^2\cdot k!.\] Output \reject if it does not hold.
\end{quote}

\begin{quote}
{\bf \tailtest$(k,T)$}: Sample a $k \times k$ matrix $X.$ Test that $|f_{k}(X)|^2
  \le T^2 k!.$ Output \reject if it does not hold. 
\end{quote}
\begin{figure*}[htb]
\begin{tabularx}{\textwidth}{|X|}
\hline
\vspace{-1mm}
{\bf Parameters}: A positive integer $n \in \N$, error parameter
$\delta \in (0,1]$, and completeness parameter $c \in (0,1]$.\\ 
{\bf Requires}: Oracle access to \stack, where $\oracle_k: \C^{k\times
k} \to \C$.
\begin{enumerate}
\item Set the following variables: $T = \nicefrac{4n}{\delta\sqrt{c}},\ d = \nicefrac{192n^2}{\delta^4c}.$
\item For each $1 \le k \le n$, 
  \begin{enumerate}
    \item Run \lintest$(n,k,\delta)$ $d$ times.
    \item Run \tailtest$(k,T)$ $d$ times.
  \end{enumerate}
\item If none of the above tests output \reject, output \texttt{Accept}.
\end{enumerate}
\\
\hline
\end{tabularx}
  \caption{The tester \ptest}
  \label{fig:ptest}
\end{figure*}
\noindent The procedure \ptest is formally defined in
Figure~\ref{fig:ptest}. 
In
the rest of the paper, we prove the following theorem about \ptest.
\begin{theorem}[Main Theorem]\label{thm:main}
  For all $n \in \N, \delta \in (0,1],$ and $c \in (0,1],$ satisfying $n
  = \Omega\left(\sqrt{\log \frac{1}{c\delta}}\right),$ given oracle
  access to \stack, where $\oracle_k : \C^{k\times k} \to \C,$ the
  procedure \ptest satisfies the following:
\begin{enumerate}
\item {\bf (Completeness)} If, for every $k \leq n,$ and every $X \in
  \C^{k \times k}$, $\Abs{\oracle_k(X) - \per_k(X)}^2 \le \delta^2 k!$, then
\ptest accepts with probability at least $1-c$.
\item {\bf (Soundness)} For every $1 \le k \le n$, either

\begin{quote}
  There exists an indicator function $1_k : \C^{k\times k} \to \{0,1\}$
  satisfying $\E_X[1_k(X)] \ge 1 - \frac{\delta^4c}{64n},$ such that,
  $\E_X [ 1_k(X)\cdot |\oracle_k(X) - \per_k(X)|^2 ] \le (2nk\delta)^2 k!.$
\end{quote}

or else,
\begin{quote}
 \ptest outputs \emph{\reject}\!\! with probability at least
$1-e^{-n}.$
\end{quote}
\item {\bf (Complexity)} The total number of queries made by \ptest is
  $O(n^4 \delta^{-4} c^{-1}).$ Moreover, assuming that each oracle
  query takes constant time, the time required by \ptest is also
  $O(n^4 \delta^{-4} c^{-1}).$
\end{enumerate}
\end{theorem}
\noindent The three parts of the theorem are proved separately in
Theorem~\ref{thm:ptester:completeness},
Theorem~\ref{thm:ptester:soundness} and
Theorem~\ref{thm:ptester:complexity} in Sections~\ref{sec:completeness},
\ref{sec:soundness} and \ref{sec:complexity} respectively.
{
\begin{remark}
  Observe that, assuming both $\nicefrac{1}{c}$ and
  $\nicefrac{1}{\delta}$ are polynomial in $n,$ the query complexity
  is $\poly(n),$ and hence, even if the oracles $\{\oracle_k\}_{k \le
    n}$ satisfy $|\oracle_k(X) - \per_k(X)|^2 \le \delta^2k!$ only with probability
  $1-\frac{1}{\poly(n)},$ \ptest would still accept with probability
  $1-c-\frac{1}{\poly(n)}.$
\end{remark}

\begin{remark} Observe that the (informal) main
  theorem~(Theorem~\ref{thm:main-informal}) stated in the introduction
  follows from Theorem~\ref{thm:main} from a simple Markov argument.
  Given $\delta = \Omega(\nfrac{1}{\poly(n)})$, set $c =
  \frac{1}{\poly(n)}$ and note that the completeness follows directly
  from Theorem~\ref{thm:main} and the previous remark.  Further, from
  the {\em Soundness} claim of Theorem~\ref{thm:main}, we
  have an indicator function $1_k : \C^{k\times k} \to \{0,1\}$
  satisfying $\E_X[1_k(X)] \ge 1 - \frac{\delta^4c}{64n} \ge 1 -
  \frac{1}{\poly(n)},$ such that, $\E_X [ 1_k(X)\cdot |\oracle_k(X) -
  \per_k(X)|^2 ] \le (2nk\delta)^2 k!  \le \poly(n)\cdot \delta^2 k!.$
 Applying Markov's inequality, we have that
$ \Pr \Brac{ 1_k(X)\cdot \Abs{\oracle_k(X) - \per_k(X)}^2 \ge \poly(n)
  \delta^2 k! } \le \nfrac{1}{\poly(n)}. $
Now, note that $1_k$ is an indicator function, and $\Pr [ 1_k(X) = 0
]$ is at most $1/\poly(n)$. This, along with the previous expression gives
that the tester outputs \reject if the sequence of oracles is not even $(\poly(n)
\cdot \delta, \nfrac{1}{\poly(n)})$-good.
\end{remark}
}

\subsection{Completeness}
\label{sec:completeness}
We first prove the completeness of \ptest: that a $(\delta, 0)$-good
sequence of oracles is accepted with probability at least $1 - c$.
\begin{theorem}[Completeness] 
\label{thm:ptester:completeness}
If, for every $k \leq n,$ and every $X \in \C^{k\times k}$,
$\Abs{\oracle_k(X) - \per_k(X)}^2 \le \delta^2 k!$, then the procedure
\ptest accepts with probability at least $1-c$.
\end{theorem}
\begin{proof}
  Suppose we are given a sequence of oracles $\{\oracle_k\}_{k \le n}$ such that for all $k \le n,$
  we have that $|\oracle_k(X) - \per_k(X) |^2 \le \delta^2 \cdot k!.$ Let
  $X$ denote a randomly sampled $k \times k$ matrix.

  We first bound the probability that the oracles \stack fail a linearity
  test. For $k=1,$ it is easy to see that
  \lintest$(n,1,\delta)$ never outputs \reject
  upon querying $\oracle_1.$ For larger $k,$ we have the following
  lemma that shows that $\oracle_k(X) \approx \sum_i
  x_{1i}\oracle_{k-1}(X_i),$ and hence {\lintest} outputs \reject only
  with small probability.
\begin{lemma}[Completeness for \lintest]
\label{lem:ptester:completeness:small-error}
For every $2 \le k \le n,$ the oracles \stack satisfy
\[\Pr_X[|\oracle_k(X) - \sum_i x_{1i}
\oracle_{k-1}(X_i)|^2 > n^2\delta^2k!] \le 2e^{-\frac{(n-1)^2}{2}}.\] 
\end{lemma}
\noindent We first complete a proof of the theorem assuming this
lemma. This lemma implies that every call to \lintest$(n,k,\delta)$ outputs \reject
with probability at most $2e^{-\frac{(n-1)^2}{2}}.$

Next, we bound the probability that the
oracles \stack fail a {\tailtest}. Using the tail bound for the
permanent given by Lemma~\ref{lem:perm-tail-bound}, we get,
$\Pr_X[|\per_k(X)| > (T-\delta)\sqrt{k!}] \le
\nfrac{(k+1)}{(T-\delta)^4}.$
Since $|\oracle_k(X) - \per_k(X) | \le \delta \cdot \sqrt{k!},$ we use it
in the above bound to get 
$\Pr_X[|\oracle_k(X)| > T\sqrt{k!}] \le
\nfrac{(k+1)}{(T-\delta)^4}.$
Thus, every call to {\tailtest} fails with probability at
most $\frac{(n+1)}{(T-\delta)^4}$. 

Now applying a union bound, we get that for $n$ that is
$\Omega\left(\sqrt{\log \frac{1}{\delta c}}\right)$, \ptest outputs
\reject with probability at most
\[\left(2e^{-\frac{(n-1)^2}{2}} +
  \frac{(n+1)}{(T-\delta)^4}\right)dn \le 384\frac{n^3}{\delta^4c}\cdot e^{-\nfrac{(n-1)^2}{2}} +
\frac{192(n+1)n^3c}{(4n - \delta^2\sqrt{c})^4} \le c.\]
\end{proof}
We now give a proof of Lemma~\ref{lem:ptester:completeness:small-error}.
\begin{proof} (\emph{of Lemma~\ref{lem:ptester:completeness:small-error}}).
We have,
\begin{align}
\left|\oracle_k(X) - \sum_i x_{1i} \oracle_{k-1}(X_i)\right| & \le \left|\oracle_k(X) - \per_k(X)\right| +
\left|\sum_i x_{1i} \per_{k-1}(X_i) - \sum_i x_{1i} \oracle_{k-1}(X_i)\right| \nonumber \\
& \le \delta\sqrt{k!} + \left|\sum_i x_{1i}(\per_{k-1}(X_i)- \oracle_{k-1}(X_i))\right|. \label{eq:ptester:completeness-error}
\end{align}
Now, since $x_{11},\ldots,x_{1k}$ are independent Gaussians with unit
variance, $\sum_i x_{1i}(\per_{k-1}(X_i)- \oracle_{k-1}(X_i))$ is a Gaussian with
variance $\sum_i |\per_{k-1}(X_i)- \oracle_{k-1}(X_i)|^2 \le k\cdot \delta^2
\cdot (k-1)! = \delta^2\cdot k!.$ Thus, the second term in
Equation~\eqref{eq:ptester:completeness-error} is bounded by
$(n-1)\delta\cdot \sqrt{k!},$ except with probability at most
$2e^{-\frac{(n-1)^2}{2}}.$ Thus, $|\oracle_k(X) - \sum_i x_{1i} \oracle_{k-1}(X_i)|
\le n\delta\cdot \sqrt{k!},$ except with probability at most
$2e^{-\frac{(n-1)^2}{2}}.$
\end{proof}

\subsection{Soundness} 
\label{sec:soundness}

The interesting part of the analysis is the soundness for \ptest,
which we prove in this section.  Given \stack, we need to define the following
indicator functions to aid our analysis:

\begin{align}
  1^{LIN}_k (X) &= \begin{cases} \I[(\oracle_k(X) - X)^2 \le
  n^2\delta^2 ], & \mbox{if } k=1\\  \I[(\oracle_k(X) - \sum_i x_{1i} \oracle_{k-1}(X_i))^2 \le
  n^2\delta^2k! ], & \mbox{if }2 \le k \le n \end{cases}
\nonumber \\
  1^{TAIL}_k (X) &= \I[\oracle_k(X)^2 \le T^2\cdot k! ],\nonumber \\
  1^{PERM}_k (X) &= \I[\per_k(X)^2 \le T^2 \cdot k! ], \nonumber \\
  1_k(X) &= 1^{LIN}_k(X) \wedge 1^{TAIL}_k(X) \wedge 1^{PERM}_k(X).
\label{eq:1k-def}
\end{align}
We now prove the following theorem.
\begin{theorem}[Soundness] 
\label{thm:ptester:soundness}
Let the indicator function $1_k$ be as
  defined by Equation~\eqref{eq:1k-def}. For every $k \leq n$, either
  both of the following two conditions hold: 
\begin{enumerate}
\item The indicator $1_k$ satisfies $\E_X[1_k(X)] \ge 1 -
  \frac{\delta^4c}{64n}.$
\item The oracle $\oracle_k$ and the indicator $1_k$ satisfy $\E_X [ 1_k(X)\cdot |\oracle_k(X) - \per_k(X)|^2 ]
  \le (2nk\delta)^2 k!,$
\end{enumerate}
or else, 
\ptest outputs \emph{\reject}\!\! with probability at least $1-e^{-n}.$
 \end{theorem}

 \begin{proof} We first prove the following lemma that shows that for all $k \le n,$ the expectation of $1_k$ is large.
\begin{lemma}[Large Expectation of $1_k$]
\label{lem:ptester:1k-complement}
Either, for every $k,$ the indicator function $1_k$ satisfies $\E_X[1_k(X)] \ge 1 - \frac{\delta^4c}{64n},$ or else, \ptest outputs \emph{\reject}\!\! with probability at least $1-e^{-n}.$
\end{lemma}
\noindent The first part of the theorem follows immediately from this lemma. The proof of this lemma is given later in this section.

For the second part of the theorem, we prove the following inductive
claim about the oracles $\{\oracle_k\}.$
\begin{lemma} (Main Induction Lemma) Suppose that for some $2 \le k \le n,$ we have, 
\label{lem:ptester:soundness:induction}
\begin{dmath*}
\E_{X \in \C^{(k-1)\times (k-1)}} [ 1_{k-1}(X)\cdot |\oracle_{k-1}(X) - \per_{k-1}(X)|^2) ] \le \eps^2_{k-1} (k-1)!,
\end{dmath*}
then, either we have,
\begin{dmath*}
\E_{X \in \C^{k\times k}}  [ 1_k(X)\cdot |\oracle_{k}(X) - \per_{k}(X)|^2) ] \le (\eps_{k-1} + 2n\delta)^2 k!,
\end{dmath*}
or else, \ptest outputs \emph{\reject}\!\! with probability at least $1-e^{-n}.$
\end{lemma}
\noindent The proof of this lemma is also presented later in this
section. Assuming this lemma, we can complete the proof of soundness
for \ptest.

For the second part of the theorem, we first show that the required bound holds for
$k=1.$ We know that for any $X \in \C,$ whenever $1_{1}(X) = 1,$ we
have $|\oracle_{1}(X)-X|^2 \le n^2 \delta^2.$ Thus,
\[\E_X [1_{1}(X)\cdot |\oracle_1(X)-\per_1(X)|^2] \le
\E_X[1_{1}^{LIN}(X)\cdot |\oracle_1(X)-X|^2] \le n^2\delta^2 < (2n\delta)^2\cdot 1!.\]
This gives us our base case. Assume that there is a $2 \le j \le n$ such that,
\[\E_{X \in \C^{(j-1)\times (j-1)}} [ 1_{j-1}(X)\cdot |\oracle_{j-1}(X) - \per_{j-1}(X)|^2 ] \le (2n(j-1)\delta)^2\cdot (j-1)!.\]
\noindent  Now, we use Lemma~\ref{lem:ptester:soundness:induction} to deduce that either,
\[\E_{X \in \C^{j\times j}} [ 1_{j}(X)\cdot |\oracle_{j}(X) - \per_{j}(X)|^2 ] \le (2nj\delta)^2\cdot j!,\]
or else, \ptest outputs \reject with probability at least $1-e^{-n}.$ Thus, by induction, either for every $ k \le n,$
\[\E_{X} [ 1_{k}(X)\cdot |\oracle_{k}(X) - \per_{k}(X)|^2] \le (2nk\delta)^2\cdot k!,\]
or else,  \ptest outputs \reject with probability at least $1-e^{-n}.$
This completes the proof of the theorem.
\end{proof}

\paragraph{Large expectation of $1_k$.}
We now prove Lemma~\ref{lem:ptester:1k-complement} that states that
the expectation of $1_k$ is large.

\begin{proof} (\emph{of Lemma~\ref{lem:ptester:1k-complement}}).
 We begin by making several claims about the structure the oracles
 \stack must have with high probability, assuming that \ptest
 accepts. First, we claim that 
 $\oracle_1$ must be close to the identity function.
\begin{claim}[Soundness of {\lintest} for $\oracle_1$]
  Either the oracle $\oracle_1$ satisfies that
\label{claim:ptester:soundness:1lin1}
\begin{align}
\label{eq:ptester:level1-lin-error}
\Pr_X\left[\left|\oracle_1(X) -
       X\right|^2 > n^2\delta^2\right] \le \frac{n}{d},
\end{align}
 or else, \ptest outputs \emph{\reject}\!\! with probability at least
 $1-e^{-n}.$
\end{claim}
\noindent A proof of this claim is included later in the section for
completeness. We also need the following two claims stating that for
every $2 \le k \le n,$ $\oracle_k(X) \approx \sum_i x_{1i}
\oracle_{k-1}(X_i)$ very often, and that $\oracle_k(X)$ does not take
large values too often.
\begin{claim}[Soundness of {\lintest}]
  Either the oracles $\{\oracle_k\}$ satisfy the following for every $2
  \le k \le n$,
\begin{align*}
\Pr_X\left[{\textstyle \left|\oracle_k(X) - \sum_i x_{1i}
\oracle_{k-1}(X)\right|^2 > n^2\delta^2k!} \right] \le \frac{n}{d},
\end{align*}
or else, \ptest outputs \emph{\reject}\!\! with probability at least $1-e^{-n}.$
\end{claim}
\begin{claim}[Soundness of \tailtest]
  Either the oracles $\{\oracle_k\}$ satisfy the following for every $k \le
  n$,
\begin{align*}
\Pr_X\left[ |\oracle_k(X)|^2 >
    T^2\cdot k! \right] \le \frac{n}{d},
\end{align*}
or else, \ptest outputs \emph{\reject}\!\! with probability at least $1-e^{-n}.$
\end{claim}
\noindent The proofs of these claims are very similar to that of
Claim~\ref{claim:ptester:soundness:1lin1} and we skip them. We can
restate the above claims in terms of $1_k^{LIN}$ and $1_k^{TAIL}$ defined
in Equation~\eqref{eq:1k-def} as follows: Either, for every $k \le n,$ 
\begin{align}
\label{eq:ptester:soundness:1k-components}
\E_X[1_k^{LIN}(X)] \ge 1 -\frac{n}{d},\ \E_X[1_k^{TAIL}(X)] \ge 1 -\frac{n}{d},
\end{align}
\noindent or else, \ptest will output \reject with probability at least $1 - e^{-n}.$

From Lemma~\ref{lem:perm-tail-bound}, we know that the
permanent does not take large values too often. To be precise,
\[\Pr_X[|\per_k(X)|^2 > T^2\cdot k!] \le
\frac{(k+1)}{T^4}.\]
Again, this implies that $\E_X[1_k^{PERM}] \ge 1 - \frac{(k+1)}{T^4}.$
Combining these three claims, we can now prove our lemma.

We know that $1_k = 1_k^{LIN} \wedge 1_k^{TAIL} \wedge 1_k^{PERM}.$ We
know that if either of the claims in
Equation~\eqref{eq:ptester:soundness:1k-components} does not hold, \ptest outputs \reject with probability at least $1-e^{-n}.$ Thus, we assume that both the claims in Equation~\eqref{eq:ptester:soundness:1k-components} hold and apply the union bound to get,
\begin{align*}
\E_X[1_k(X)] & \ge 1 - \E_X[1-1_k^{LIN}(X)] - \E_X[1-1_k^{TAIL}(X)]- \E_X[1-1_k^{PERM}(X)] \\
& \ge 1 - \frac{n}{d} -\frac{n}{d}- \frac{k+1}{T^4} \ge
1-\frac{\delta^4c}{96 n} - \frac{(n+1)\delta^4c^2}{256n^{4}} \ge
1-\frac{\delta^4c}{64n},
\end{align*}
for large enough $n.$
\end{proof}

\paragraph{Main Induction Lemma.} We now give a proof of the main induction lemma.
\begin{proof} (\emph{of Lemma~\ref{lem:ptester:soundness:induction}}).
  Recall that $X_i$ is the minor obtained by deleting the first row
  and the $i^{\textrm{th}}$ column from $X.$  We first split the
  probability space for $X \in \C^{k\times k}$ according to whether
  all of its minors $X_i$ satisfy $1_{k-1}(X_i) = 1$ or not.
\begin{equation}\label{eq:ptester:split-by-minors}
\begin{aligned}
\norm{ 1_k(X) (\oracle_k(X) - \per_k(X))}^2 &=  \overbrace{\norm{1_k(X) \prod_i 1_{k-1}(X_i) (\oracle_k(X) -  \per_k(X))}^2}^{\term{trm:good}} \\
&+ \underbrace{\norm { 1_k(X) (1-\prod_i 1_{k-1}(X_i)) (\oracle_k(X) - \per_k(X))}^2}_{\term{trm:bad}}
\end{aligned}
\end{equation}
Let $\tilde{1}_k(X) = 1_k(X) \prod_i 1_{k-1}(X_i).$ \Cref{trm:good},
above, is bounded by adding and subtracting the expression $\sum_i
x_{1i} \oracle_{k-1}(X_i)$ and then expanding the permanent along the
first row. 
\begin{equation}\label{eq:ptester:row-expansion}
\begin{aligned}
   \norm{\tilde{1}_k(X) (\oracle_k(X) - \per_k(X) ) } &\le 
   \underbrace{\norm{\tilde{1}_k(X) [ \oracle_k(X) - \sum_i x_{1i}
     \oracle_{k-1}(X_i) ]}}_{\term{trm:lin}}\\
   &+ \underbrace{\norm{ \tilde{1}_k(X) [\sum_i x_{1i}\oracle_{k-1}(X_i) - \sum_i
     x_{1i}\per_{k-1}(X_i)]}}_{\term{trm:exp}}
  \end{aligned}
\end{equation}
We know that for all $X$ such that $1_k(X) = 1,$ $|\oracle_k(X) -
\sum_i x_{1i}\oracle_{k-1}(X_i)|^2$ is bounded by $n^2 \delta^2 k!.$
Thus, \cref{trm:lin} in eq.~\eqref{eq:ptester:row-expansion} is at most
$n \delta \sqrt{k!}.$
\begin{align}
\eqref{trm:lin} \le %
  \norm{1_k(X) (\oracle_k(X) - \sum_i x_{1i}\oracle_{k-1}(X_i))}%
    \le n\delta \sqrt{k!}  \label{eq:ptester:sq-error:first-term}
\end{align}
\Cref{trm:exp} is bounded by using the induction assumption:
\begin{equation}
\begin{aligned}
\eqref{trm:exp}^2 = \BNorm{1_k(X) &\prod_i 1_{k-1}(X_i) \BBrac{\sum_i x_{1i}\oracle_{k-1}(X_i) - \sum_i
  x_{1i}\per_{k-1}(X_i)}}^2  \\
&    \le \E_{X_1,\ldots X_k} \E_{x_{11}, \ldots ,x_{1k}}
  \left[\prod_i 1_{k-1}(X_i)\cdot \left|\sum_i x_{1i}\oracle_{k-1}(X_i) - \sum_i x_{1i}\per_{k-1}(X_i)
  \right|^2\right]  \\
  &  \le \E_{X_1,\ldots X_k}\left[\prod_i 1_{k-1}(X_i)\cdot \sum_i |\oracle_{k-1}(X_i) -
  \per_{k-1}(X_i)|^2\right] \\
  &   \le \sum_i \E_{X_i}\left[ 1_{k-1}(X_i) \cdot |\oracle_{k-1}(X_i) -
    \per_{k-1}(X_i)|^2\right] \\
  &   \le k \eps_{k-1}^2 (k-1)! = \eps_{k-1}^2 k!
\end{aligned}\label{eq:ptester:sq-error:second-term}
\end{equation}
Combining
eqs.~\eqref{eq:ptester:row-expansion}, \eqref{eq:ptester:sq-error:first-term},
and~\eqref{eq:ptester:sq-error:second-term}, we get,
\begin{equation}
 \eqref{trm:good} = \E_X\left[1_k(X) \prod_i 1_{k-1}(X_i)\cdot  \left|\oracle_k(X) - \per_k(X)\right|^2 \right] \le (\eps_{k-1}+n\delta)^2\cdot k!.
\label{eq:ptester:sq-error:third-term}
\end{equation}
Next, we bound \cref{trm:bad} as follows.  First use
\cref{lem:ptester:1k-complement} to deduce $\Pr_X [ 1_{k-1}(X_i) =
0]\le \frac{\delta^4c}{64n}$ (If it does not hold, we know that \ptest
outputs \reject with probability at least $1-e^{-n}$). Whenever
$1_k(X) = 1,$ we have $|\oracle_k(X)| \le T\sqrt{k!}$ and $|\per_k(X)|
\le T\sqrt{k!}.$ This implies that $1_k(X)\cdot |\oracle_k(X) -
\per_k(X)|^2 \le 4T^2 k! $ everywhere. Thus, we have,
\begin{equation}\begin{aligned}
\eqref{trm:bad} &= \norm { 1_k(X) (1-\prod_i 1_{k-1}(X_i))
  (\oracle_k(X) - \per_k(X))}^2 \le 4T^2 k! \E_X\left[1 - \prod_i
1_{k-1}(X_i)\right]\\
&\le 4T^2 k! \E_X \left[\sum_i (1 - 1_{k-1}(X_i) )\right]\\
& \le 4T^2 k! \cdot k \cdot \frac{\delta^4c}{64n} \le n^2\delta^2\cdot k!.  \label{eq:ptester:sq-error:fourth-term}
\end{aligned}\end{equation}
Combining eqs.~\eqref{eq:ptester:split-by-minors}, \eqref{eq:ptester:sq-error:third-term},
 and~\eqref{eq:ptester:sq-error:fourth-term} completes the proof:
\[ \E \left[1_k(X)\cdot |\oracle_k(X) - \per_k(X)|^2\right] \le \left(\left(\eps_{k-1}
    + n\delta\right)^2 + n^2 \delta^2\right)\cdot k! \le
\left(\eps_{k-1} + 2n\delta\right)^2\cdot k!. \qedhere\]
\end{proof}

\paragraph{Proof of Claim~\ref{claim:ptester:soundness:1lin1}}
For completeness, we include a proof of Claim~\ref{claim:ptester:soundness:1lin1}.
\begin{proof} 
  Assume that the oracle $\oracle_1$ does not satisfy
  Equation~\eqref{eq:ptester:level1-lin-error}. We know that
  $\left|\oracle_1(X) - X\right|^2 > n^2\delta^2$ iff
  \lintest$(n,1,\delta)$ outputs \reject when
  $X$ is sampled by the procedure. Thus, the measure of points $X
  \in \C$ that would fail the test
  \lintest$(n,1,\delta)$ is at least $\nfrac{n}{d}.$
  This implies that the probability that none of the $d$ calls to
  \lintest$(n,1,\delta)$ made by \ptest output
  \reject is at most $\left(1-\nfrac{n}{d}\right)^d \le
  e^{-n}.$
\end{proof}

\subsection{Complexity}
\label{sec:complexity}
We finally note that the complexity of \ptest is polynomially
bounded in the input parameters.
\begin{theorem}[Query and Time Complexity] 
\label{thm:ptester:complexity}
The total number of queries
  made by \ptest to all the oracles is $O(n^2d) = O(n^4 \delta^{-4}
  c^{-1}).$ Moreover, assuming that each oracle query takes constant
  time, the time required by \ptest is also $O(n^2d) = O(n^4 \delta^{-4}
  c^{-1}).$
\end{theorem}
\begin{proof}
  By the definition of \ptest, it makes $dn$ calls to {\lintest} and
  $dn$ calls to {\tailtest}. Each call to {\lintest} with parameters
  $n,k,\delta,$ makes at most $k+1$ queries to the oracles (for $k=1,$
  it makes only one query), and requires $O(k)$ time. Each call to
  {\tailtest} makes 1 query and requires $O(1)$ time.  Thus, the total
  number of queries made is $O(dn^2) = O(n^4 \delta^{-4} c^{-1}),$ and
  the total time required is also $O(dn^2) = O(n^4 \delta^{-4}
  c^{-1}).$
\end{proof}
\noindent Thus, if $\nfrac{1}{\delta}$ and $\nfrac{1}{c}$ are $\poly(n),$
the query complexity of \ptest is also $\poly(n).$

\ignore{
\subsection{Testing Simulator for BQC}

Finally, we combine the tester with the reduction in AA to the
permanent function to obtain a tester for a simulator $\BPP^\NP$.  

If $\norm{D(A) - D_\oracle(A)} \le \beta$, for every circuit
description $A$, then there is a $\BPP^{\NP^\oracle}$ 
}


\ignore{\section{Conclusion} 

We  study  and  design  a  tester for  oracles  that  approximate  the
permanent  of matrices  from the  gaussian ensemble.   A next  step in
understanding the complexity of the bosonic quantum system would be to
design a tester  directly for the simulator of  such a quantum system.
Perhaps our techniques might prove useful.}

\section{Acknowledgments}

The authors  would like to  thank Madhur Tulsiani and  Rishi Saket for
extensive discussions during early stages of this work.  We would also
like  to  thank  Scott Aaronson, Alex Arkhipov,  Swastik Kopparty  and
Srikanth Srinivasan for helpful discussions.

\bibliographystyle{alpha}
\bibliography{papers}

\appendix
\section{Remaining proofs}
\begin{claim}[Soundness of {\lintest}]
  Either the oracles $\{\oracle_k\}$ satisfy the following for every $2
  \le k \le n$,
\begin{align}
\label{eq:ptester:level-k-lin-error}
\Pr_X\left[{\textstyle \left|\oracle_k(X) - \sum_i x_{1i}
\oracle_{k-1}(X)\right|^2 > n^2\delta^2k!} \right] \le \frac{n}{d},
\end{align}
or else, \ptest outputs \emph{\reject}\!\! with probability at least $1-e^{-n}.$
\end{claim}
\begin{proof}
  Assume that there exists a $k,$ such that $2 \le k \le n$ and the
  oracle $\oracle_k$ does not satisfy
  Equation~\eqref{eq:ptester:level-k-lin-error}. We recall that
  \lintest$(n,k,\delta)$ outputs \reject iff
  the sampled $X \in \C^{k\times k}$ satisfies $\left|\oracle_k(X) - \sum_i x_{1i}
    \oracle_{k-1}(X)\right|^2 > n^2\delta^2k!.$ Thus, a randomly sampled $X$
  will fail \lintest$(n,k,\delta)$ with probability at
  least $\nfrac{n}{d}.$ This implies that the probability that none of
  the $d$ calls made by
  \ptest  to \lintest$(n,k,\delta)$  output \reject is at most
  $\left(1-\nfrac{n}{d}\right)^n \le e^{-n}.$
\end{proof}
\begin{claim}[Soundness of \tailtest]
  Either the oracles $\{\oracle_k\}$ satisfy the following for every $k \le
  n$,
\begin{align}
\label{eq:ptester:tail-bound}
\Pr_X\left[ |\oracle_k(X)|^2 >
    T^2\cdot k! \right] \le \frac{n}{d},
\end{align}
or else, \ptest outputs \emph{\reject}\!\! with probability at least $1-e^{-n}.$
\end{claim}
\begin{proof}
Suppose for some $k \le n,$ the oracle $\oracle_k$ does not satisfy
Equation~\ref{eq:ptester:tail-bound}. Thus, for this choice of $k,$
the test \tailtest$(k,T)$ fails with probability at least
$\nfrac{n}{d}.$ This implies that the probability that at least one of
the $d$ calls by \ptest to \tailtest$(k,T)$ outputs
\reject with probability at least $1 - (1-\nfrac{n}{d})^d \ge 1-e^{-n}.$
\end{proof}


\end{document}